\pdfoutput=1
\documentclass[twoside]{article}
\usepackage{ukai}

\usepackage{amsthm}
\usepackage{booktabs}
\usepackage{tikz}
\usepackage{pgfplots}
\usetikzlibrary{arrows.meta, positioning, fit, backgrounds, calc, shadows.blur}
\pgfplotsset{compat=1.16}
\usepackage{microtype}
\usepackage{colortbl}
\usepackage[most]{tcolorbox}

\definecolor{ink}{HTML}{2B2F36}
\definecolor{paperbg}{HTML}{FBF8F1}
\definecolor{plotaccent}{HTML}{7A5C2E}
\definecolor{cardline}{HTML}{B7BECA}
\definecolor{cSafe}{HTML}{2E7D5B}
\definecolor{cDanger}{HTML}{B23A48}
\definecolor{cProc}{HTML}{3F6CA6}
\definecolor{cWarn}{HTML}{A9781C}

\renewcommand{\confyear}{2026}

\titlespacing*{\section}{0pt}{2.2ex plus .3ex minus .2ex}{1.1ex plus .2ex}
\newtheoremstyle{propstyle}{6pt}{6pt}{\itshape}{0pt}{\bfseries\color{cSafe}}{.}{.6em}{}
\newtheoremstyle{defnstyle}{6pt}{6pt}{\normalfont}{0pt}{\bfseries\color{cProc}}{.}{.6em}{}
\theoremstyle{propstyle}
\newtheorem{theorem}{Theorem}
\newtheorem{proposition}{Proposition}
\theoremstyle{defnstyle}
\newtheorem{definition}{Definition}
\tcolorboxenvironment{definition}{blanker, breakable, before skip=6pt, after skip=6pt,
  left=9pt, right=1pt, top=1pt, bottom=1pt, borderline west={2pt}{0pt}{cProc!85}}
\tcolorboxenvironment{proposition}{blanker, breakable, before skip=6pt, after skip=6pt,
  left=9pt, right=1pt, top=1pt, bottom=1pt, borderline west={2pt}{0pt}{cSafe!85}}

\hypersetup{hidelinks,
  pdftitle={Certified Multi-Source Integrity for Structured Agent Actions},
  pdfauthor={Anmol Pandey, Aditya Jain, Liang Chen, Carsten Maple, Christo Panchev},
  pdfkeywords={LLM agents, prompt injection, trustworthy AI, certified robustness, AI systems security},
  pdfsubject={UK AI Conference 2026}}
\newcommand{\exec}{\textsc{Execute}}
\newcommand{\abst}{\textsc{Abstain}}

\title{\customtitle{Certified Multi-Source Integrity for Structured Agent Actions}}

\setauthorsshort{Pandey et al.}

\author{
    Anmol Pandey$^1$, Aditya Jain$^1$, Liang Chen$^2$, Carsten Maple$^1$, Christo Panchev$^1$ \\[2pt]
    {$^1$University of Warwick, $^2$University of Hertfordshire} \\[1pt]
    {\small \texttt{Anmol.Pandey@warwick.ac.uk, aditya.jain.5@warwick.ac.uk, l.chen26@herts.ac.uk}} \\[1pt]
    {\small \texttt{cm@warwick.ac.uk, christo.panchev@warwick.ac.uk}} \\[2pt]
    {\small Code and data: \url{https://github.com/anmolpandey299/certified-multisource-integrity}}
}

\date{}

\usepackage{float}
\begin{document}
\maketitle

\begin{abstract}
LLM agents increasingly take privileged, often irreversible structured actions, such as paying an invoice. They assemble each action from action-critical fields in documents and tool outputs
that an adversary can corrupt, and indirect prompt injection can drive the model itself to extract
attacker-chosen values. Current defenses gate on a source's trust label or certify free-text answer quality. None certifies the integrity of a coupled, policy-bound structured action under a
corruption budget that accounts for shared upstream sources. We characterize when such an action is safely
certifiable and give the maximally live safe certifier. It admits an action only when each field clears the
rule its evidence structure supports: a bounded corruption radius over \emph{corruption-distinct} evidence
classes, counted by a minimum hitting set so that re-publishing or laundered copies cannot manufacture
a quorum, deterministic reconciliation for complementary fields, and a trusted anchor where the evidence leaves
a field single-sourced. We formalize two robustness notions, validate each mechanism by ablation, and measure
how often the multi-source precondition holds on sanctions designations ($70{,}966$ entities) and software
supply-chain provenance ($450$ packages). Under upper-bound proxies, genuine corroboration is a minority
phenomenon in both, and naive attestation counting overstates it, since witnesses that look independent
collapse to two corruption-distinct domains once shared origin is counted. Across five current models in a
real agent loop, a realistic injection fools every model but one and a naive agent then executes the
fraudulent action on most attacks. The certifier admits no unsafe action and recovers the correct value where
corroboration permits, while action-gating and provenance baselines are broken in every world of our harness by some attack in its
space.
\end{abstract}


\section{Introduction}\label{sec:intro}

Large language model (LLM) agents are increasingly trusted to take privileged structured actions on a
user's behalf, such as paying an invoice, screening a counterparty against sanctions, or filing a record.
For a payment the agent must settle on a payee, an amount, and an account, each read from documents and
tool outputs. In realistic deployments this evidence is
multi-source and adversarially corruptible: a purchase order, an invoice, a bank verification response, and a
sanctions feed originate from different parties, and an attacker may influence any of them. Indirect
prompt injection makes this sharper, since a tampered source can induce the model itself to extract an
attacker-chosen value \cite{agentdojo}, so the agent proposes an action that is wrong yet well-formed. Since
these actions are often irreversible, one fooled extraction can send a fraudulent payment.

Existing defenses do not address this directly. Information-flow systems and policy monitors gate an action
on the trust label or provenance of its inputs. That answers whether data may flow to an action, not whether
the value in that data is correct, so a decisive value from an untrusted source forces a blanket refusal.
Certified-robustness work such as RobustRAG~\cite{robustrag} bounds the quality of a free-text answer against
injected passages. It certifies an answer rather than an action, and it counts passages, so duplicated copies
inflate the budget. None of these certifies the integrity of a structured action under a corruption budget
that accounts for shared upstream sources, and none measures whether the corroboration such a certificate
needs is present in real evidence. Section~\ref{sec:related} places the present work against these systems.

We study certified multi-source integrity for structured agent actions. The certifier reads each source in
isolation and admits an action only when every field independently clears the rule that suits its evidence
structure. At the center is a bounded corruption radius over corruption-distinct evidence classes. Rather
than counting raw attestations, the certifier counts the smallest number of independent control domains an
attacker must corrupt to erase a fact, a minimum hitting set, so that authorities that merely republish or
transpose one another are not mistaken for independent witnesses. A claim-atomic discipline then counts
votes once per authenticated control domain, which prevents an attacker from laundering copies or issuing several records to inflate a single
corrupted source into an apparent quorum. Fields that lack redundant attestation are handled honestly
rather than forced into a vote. A complementary field such as an invoiced amount is checked against an
authorized ceiling by deterministic reconciliation, and a single-sourced field such as the payee account
is gated by a trusted anchor (Figure~\ref{fig:arch}). We keep two robustness notions apart. Exact-value tolerance asks the certificate to survive the adversary
intact, safety with abstention only that it never be wrong, and conflating them overstates the guarantee.
Section~\ref{sec:defs} states both precisely.

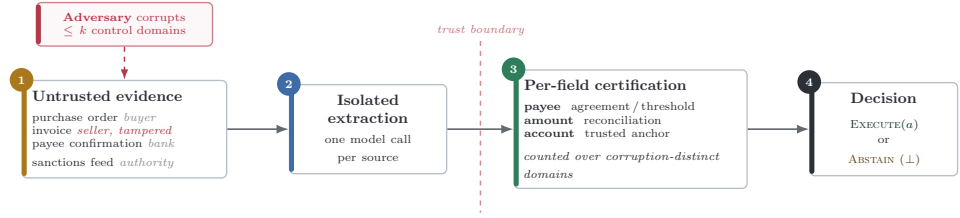
\begin{figure}[t]
\centering
\resizebox{0.76\linewidth}{!}{%
\begin{tikzpicture}[
  font=\small, >={Latex[length=2.4mm]},
  card/.style={draw=cardline, line width=0.8pt, rounded corners=3pt, fill=white, align=left,
    inner sep=6pt, text=ink},
  cardc/.style={draw=cardline, line width=0.8pt, rounded corners=3pt, fill=white, align=center,
    inner sep=6pt, text=ink},
  ar/.style={-{Latex[length=2.6mm]}, line width=1pt, ink!72},
  badge/.style={circle, fill=#1, text=white, font=\bfseries\scriptsize, inner sep=0pt, minimum size=5mm}
]
  % Card 1: sources, one tampered (red)
  \node[card, text width=39mm] (src) {\textbf{Untrusted evidence}\\[4pt]
    \scriptsize purchase order \textcolor{ink!55}{\itshape buyer}\\
    \scriptsize invoice \textcolor{cDanger}{\itshape seller, tampered}\\
    \scriptsize payee confirmation \textcolor{ink!55}{\itshape bank}\\
    \scriptsize sanctions feed \textcolor{ink!55}{\itshape authority}};
  % Card 2
  \node[cardc, text width=29mm, right=13mm of src] (ext) {\textbf{Isolated}\\\textbf{extraction}\\[4pt]
    \scriptsize one model call\\ per source};
  % Card 3: per-field rules, scannable
  \node[card, text width=45mm, right=14mm of ext] (cert) {\textbf{Per-field certification}\\[4pt]
    \scriptsize \textbf{payee}\ \ agreement\,/\,threshold\\
    \scriptsize \textbf{amount}\ \ reconciliation\\
    \scriptsize \textbf{account}\ \ trusted anchor\\[3pt]
    \scriptsize\textit{counted over corruption-distinct domains}};
  % Card 4: color-coded decision
  \node[cardc, text width=27mm, right=13mm of cert] (dec) {\textbf{Decision}\\[4pt]
    \textcolor{cSafe!42!black}{\scriptsize\exec$(a)$}\\[1pt] \scriptsize or\\[1pt]
    \textcolor{cWarn!55!black}{\scriptsize\abst\ $(\bot)$}};
  % left accent bars
  \foreach \c/\col in {src/cWarn, ext/cProc, cert/cSafe, dec/ink}{%
    \draw[\col, line width=3pt, line cap=round]
      ([xshift=1.6pt,yshift=-3pt]\c.north west) -- ([xshift=1.6pt,yshift=3pt]\c.south west);}
  % numbered stage badges
  \node[badge=cWarn] at (src.north west) {1};
  \node[badge=cProc] at (ext.north west) {2};
  \node[badge=cSafe] at (cert.north west) {3};
  \node[badge=ink]   at (dec.north west) {4};
  % flow
  \draw[ar] (src) -- (ext);
  \draw[ar] (ext) -- (cert);
  \draw[ar] (cert) -- (dec);
  % trust boundary between extraction and certification
  \coordinate (tbx) at ($(ext.east)!0.5!(cert.west)$);
  \draw[cDanger!58, dashed, line width=0.9pt]
    ($(tbx|-cert.north)+(0,6mm)$) -- ($(tbx|-cert.south)-(0,5mm)$);
  \node[font=\scriptsize\itshape, text=cDanger!72, anchor=south] at ($(tbx|-cert.north)+(0,6mm)$)
    {trust boundary};
  % adversary
  \node[cardc, draw=cDanger!65, fill=cDanger!4, text=cDanger, text width=33mm, font=\scriptsize,
    above=7mm of src] (adv) {\textbf{Adversary}\ corrupts $\le k$ control domains};
  \draw[cDanger, line width=3pt, line cap=round]
    ([xshift=1.6pt,yshift=-3pt]adv.north west) -- ([xshift=1.6pt,yshift=3pt]adv.south west);
  \draw[-{Latex[length=2.2mm]}, line width=0.9pt, cDanger, dashed] (adv.south) -- (src.north);
\end{tikzpicture}}
\caption{The certifier as an admission gate. Each source is read in isolation, every action-critical field is
decided by the rule its evidence structure supports and counted over corruption-distinct domains, and the
action executes only if every field is admitted, otherwise it abstains. An adversary within budget $k$
corrupts at most $k$ control domains, here the tampered invoice.}
\label{fig:arch}
\end{figure}

\noindent This paper makes four contributions.
\begin{itemize}\setlength{\itemsep}{2pt}\setlength{\topsep}{2pt}\setlength{\parskip}{0pt}
  \item A characterization of when a structured action is safely certifiable, with the corruption-distinct
  count as a minimum hitting set. Execution is sound exactly when the budgeted feasible set is a singleton,
  and the certifier meeting this bound is \emph{maximally live}. The Byzantine safety and exact-value radii,
  and the invariance to laundering, follow as specializations (Section~\ref{sec:defs}).
  \item A controlled ablation showing that each mechanism is load-bearing. Removing the transaction join
  key, the trusted anchor, the mandatory-source rule, or claim-atomic counting each opens exactly one
  attack family, while the full certifier admits none (Section~\ref{sec:mech}).
  \item A measurement of where that precondition actually holds. On $70{,}966$ sanctions designations only
  about $15\%$ carry three or more issuing domains. On $450$ npm packages an apparent third witness
  disappears entirely once its shared origin is counted, and no package reaches three. Driven by the
  witnesses that remain, the certifier catches a single-domain substitution on every corroborated package
  (Section~\ref{sec:scope}).
  \item A real-model end-to-end evaluation over five current models. Only one resisted the injection here, and the
  naive agent then failed on most attacks, yet the certificate held throughout. In this evaluation it outperforms
  faithful capability-gating and policy-monitor defenses on safety and utility
  at once, and pays even the legitimate novel payees that allowlisting must refuse
  (Sections~\ref{sec:mitl} and \ref{sec:baseline}).
\end{itemize}

We are deliberate about scope. Real evidence corroborates only a minority of relations, the account field
rests on a trusted anchor rather than on the radius, and the adversary in the live-model study is
property-based rather than optimizing. The security guarantee is the theorem of Section~\ref{sec:defs}, and the experiments show
the mechanism behaving as designed wherever its preconditions hold.

\section{Threat Model and Definitions}\label{sec:defs}

We model an agent that proposes a privileged structured action and a certifier that decides whether to
execute it. The model is deliberately spare, so that the guarantees rest on stated assumptions rather than
on implementation detail.

\paragraph{System and adversary.}
An agent proposes a privileged structured action $a=(a_1,\dots,a_d)$, a tuple of action-critical fields.
For a business payment these are the payee identity, the payable amount, and the payee account. A certifier
$M$ decides each field independently, and the action executes only when every field is admitted. Each field
is supported by a set of attestations. An attestation is a pair $(v,c)$ in which an evidence class $c$
asserts a value $v$. An evidence class is an origin of attestation with its own corruption mode, for
instance a buyer document, a seller document, a bank verification response, or a designation authority.
Every attestation is bound, by an unforgeable origin identifier, to the control domain that produced it, and
carries a provenance root. Independent records carry distinct roots, and a copy inherits the root of its
source, which is what later lets the certifier tell a genuine second opinion from a laundered duplicate.

The adversary selects a set $C\subseteq U$ of control domains subject to a corruption budget $|C|\le k$,
where $U$ is the universe of control domains. Corrupting $C$ affects every attestation whose dependency set
meets $C$. Within its reach the adversary may set values arbitrarily and emit arbitrarily many records,
including several distinct original records as well as laundered copies attributed to roots it controls.
Elsewhere attestations are honest and report the field's true value $f(x)$ after deterministic
canonicalization. The adversary cannot forge an attestation bound to the origin identifier of a control
domain it does not control, an assumption we make explicit below. The adversary knows the
certifier, the corpus, and the policy, and may aim either to make the certifier execute a value other than
the truth or to force it to abstain.

\paragraph{Trusted obligation skeleton.}
The certifier protects values within a fixed object, not the choice of object. From trusted user intent
alone, before any untrusted evidence is read, we fix an obligation skeleton $\Gamma$ pinning down the
operation, transaction key, mandatory relations, eligible sources, allowed destinations, amount cap, and
policy version. Untrusted evidence may only fill the values $\Gamma$ leaves open. It may not redefine the operation, the
transaction the action settles, or which sources are eligible, and an action is admissible only if it
instantiates $\Gamma$. This closes a selection attack that field certification alone
leaves open. A poisoned source steering the agent toward a different transaction yields a certificate sound
for the wrong object, rejected at the skeleton before any vote is counted.

\paragraph{Corruption-distinctness.}
Apparent independence is not independence under corruption. An authority that republishes or transposes
another shares its corruption mode, so two lists that copy a single upstream are one source for an attack.
The concrete case is a software package, whose build provenance and repository back-reference look like two
witnesses to its source but are both GitHub artifacts, so one compromise forges both and the apparent pair is
one domain (Figure~\ref{fig:supply}). We therefore measure redundancy by the number of independent corruptions needed to
erase the evidence, not by the attestations on display.

\begin{definition}[Dependency set]
For an attestation $att$, let $D(att)$ be the set of control domains for which corrupting any single member
removes $att$. An authority that designates an entity on its own has $D(att)=\{\text{that authority}\}$. A
national listing that merely transposes a United Nations designation has
$D(att)=\{\text{the nation},\ \text{the UN}\}$, because corrupting either erases it.
\end{definition}

\begin{definition}[Corruption-distinct count]\label{def:cdcount}
For a fact attested by $\{att_i\}$, the corruption-distinct count is the size of the smallest set of control
domains that meets every attestation's dependency set,
\begin{equation}
  m \;=\; \min\bigl\{\,|X| : X\cap D(att_i)\neq\emptyset \ \text{for all } i \,\bigr\}.
\end{equation}
Equivalently, $m$ is the fewest control domains an adversary must corrupt to remove every attestation of the
fact. It is a minimum hitting set, and it reduces to minimum vertex cover when each dependency set has two
members. For a threshold predicate such as ``this entity is designated'', the presence-removal radius is
$m-1$, since the adversary needs $m$ corruptions to flip the predicate and one honest attestation survives
any $m-1$ of them.
\end{definition}

\noindent Section~\ref{sec:scope} measures an upper-bound proxy for $m$ on real data, since exact
dependency sets are not yet available.

\begin{definition}[Budgeted influence]
For a budget of $k$ control domains, let the affected set of $X\subseteq U$ be
$A(X)=\{i:D(att_i)\cap X\ne\emptyset\}$, and let $b_k=\max_{|X|\le k}|A(X)|$ be the largest number of
attestations any $k$ domains can alter. When the attestations of a fact are corruption-distinct their
dependency sets are disjoint and $b_k=k$, so one corruption moves one vote. When domains are shared a single
corruption can move several votes and $b_k>k$. The propositions below are stated for corruption-distinct
classes, where $b_k=k$. In the general case the same arguments hold with $k$ replaced by $b_k$, so the
exact-value condition becomes $N>2b_k$ and the safety radius is $m-1$.
\end{definition}

\paragraph{Certifier rules.}
The certifier maps the observed attestations for a field to either $\exec(v)$ or $\bot$, an abstention that
escalates to a human. Different field structures call for different rules, and forcing every field through a
single rule is one of the mistakes the design avoids. \emph{Agreement or escalate} governs multi-source
identity fields: execute $v$ when a mandatory quorum of corruption-distinct classes is present and all of
them agree on $v$ after canonicalization, and abstain otherwise. \emph{Unique threshold by count} governs
multi-source value fields: execute a value when at most $k$ corruption-distinct classes dissent from it,
equivalently when its support is at least $N-k$, \emph{and it is the only value to clear that threshold}, which we call unique-threshold admission.
Uniqueness is load-bearing: at $N\le 2k$ an adversary can make a challenger reach $N-k$ as well, so two
values are feasible and admitting the larger count would certify a wrong one. \emph{Reconciliation}
governs complementary fields by a deterministic constraint, for example the invoiced amount lying within the
authorized order plus tolerance, a check rather than a vote. Reconciliation bounds the payment by what the
buyer authorized rather than certifying the invoiced figure, so that field carries a bound, not a radius. \emph{The anchor gate} governs fields the
evidence flow leaves single-sourced, such as the payee account: execute when the proposed account both
matches the payee name at the bank and appears on the onboarding allowlist, and abstain otherwise. Safety
here rests on trusted anchors rather than a corruption radius, a distinction we preserve.

\paragraph{Two robustness notions.}
We separate two guarantees that are easy to conflate and whose conflation flatters the result.
\begin{definition}[Exact-value tolerance]
$M$ is $k$-exact-value-tolerant when $M_C(x)=f(x)$ for every $|C|\le k$, so the certified value is unchanged
by any corruption within budget.
\end{definition}
\begin{definition}[Safety with abstention]
$M$ is $k$-safe-with-abstention when $M_C(x)\in\{f(x),\bot\}$ for every $|C|\le k$, so the certifier returns
the correct value or refuses, and never a wrong one.
\end{definition}
\noindent The second is weaker, giving up availability under attack, but reachable with far less redundancy,
which matters since real multi-source coverage is scarce (Section~\ref{sec:scope}).

\paragraph{When safe execution is possible.}
Fix the observed attestations, with values $x_i$, and the trusted skeleton $\Gamma$ of join key, mandatory
sources, amount cap, and destination constraints, all settled before untrusted evidence is read. Write
$\mathrm{Dis}_j(a)=\{i:x_i\neq a_j\}$ for the attestations of field $j$ that disagree with the value $a_j$.

\begin{definition}[Feasible action]\label{def:feasible}
An action $a$ is $k$-feasible given the observation if it satisfies $\Gamma$ and one corruption set of at most
$k$ domains explains every disagreeing attestation across all fields at once, that is
$\tau\bigl(\bigcup_j \mathrm{Dis}_j(a)\bigr)\le k$ for the minimum hitting set $\tau$ over those attestations'
dependency sets. Let $F_k$ be the set of $k$-feasible actions.
\end{definition}

\begin{theorem}[Safe execution is exactly singleton feasibility]\label{thm:feasible}
A certifier that never executes a wrong action under any $|C|\le k$ may execute $a$ if and only if
$F_k=\{a\}$, and when $|F_k|\neq1$ every $k$-safe certifier must abstain. The certifier that executes the
unique feasible action and abstains otherwise is $k$-safe and \emph{maximally live}: its abstention set lies
inside that of every $k$-safe certifier.
\end{theorem}

\noindent The proof is an indistinguishability argument (Appendix~\ref{app:proofs}): two feasible actions
yield one observation from two within-budget worlds, so committing to either is wrong in the other, while a
unique feasible action is the truth in every allowed world. Whenever this certifier abstains two worlds are
genuinely indistinguishable, so ``it abstains too often'' is not an objection without enlarging $k$. The
budget couples the fields, so feasibility is joint, not per-field. A single-sourced field has two feasible
values for any $k\ge1$, so no certifier executes it without a trusted anchor: the anchor is necessary.
The rules below decide each field separately, a sound specialization: unique per-field admission implies the
singleton the theorem requires, so the composition is $k$-safe, though it can abstain where the joint rule
certifies (Proposition~\ref{prop:compose}).

\begin{proposition}[Agreement or escalate, the conservative case]\label{prop:safe}
Under agreement-or-escalate over $N$ corruption-distinct classes, if at least one class is honest then $M$ is
$(N-1)$-safe-with-abstention.
\end{proposition}
\noindent Unanimity makes the agreed value the only feasible one once $N>k$. It is safe but not maximally live, since it
abstains where the threshold rule still certifies.

\begin{proposition}[Unique-threshold admission is maximally live]\label{prop:exact}
Admit a value when at most $k$ corruption-distinct classes dissent, equivalently when its support
is at least $N-k$, and it is the unique value clearing that threshold. Then any admitted value equals $f(x)$, and $N$ agreeing classes are certifiable against
budget $k$ if and only if $N>2k$, the exact-value radius $\lfloor (N-1)/2\rfloor$.
\end{proposition}
\noindent For disjoint domains this is Theorem~\ref{thm:feasible} made concrete, the singleton condition being
$N>2k$. The threshold is the standard Byzantine bound, here applied per predicate over a dependency-aware
electorate with vote identity preserved, and shown \emph{optimal} rather than one design among many.

\begin{proposition}[Vote identity bounds adversarial weight]\label{prop:launder}
When each attestation is attributed to its authenticated originating control domain and a domain casts at
most one vote per predicate, an adversary controlling $k$ control domains contributes at most $k$ votes,
regardless of how many distinct original records or derived copies it emits.
\end{proposition}
\noindent The bound is immediate: the votes for a predicate are the distinct controlling domains, and the
adversary controls $|C|\le k$ of them. The subtlety is what attribution requires. Counting once per
provenance root is not enough. Root-counting collapses laundered \emph{copies} onto the root they inherit,
but a single compromised authority can also issue several distinct \emph{original} records, each with its own
root, so root-counting would hand one domain several votes and break the radius. Authentication binds every
record, original or copied, to the domain that produced it, and the certifier counts domains rather than
roots. Without this, $t$ copies or $t$ original records from one corrupted source contribute $t$ votes and
manufacture a quorum from a single corruption, the attack the domain rule removes and the counting experiment
in Section~\ref{sec:mech} isolates.

\paragraph{Trusted computing base.}
These guarantees hold relative to four assumptions, which we state explicitly. First, authentication.
Each attestation is bound to an unforgeable origin identifier, so the adversary cannot fabricate
attestations from classes it does not control, and resistance to Sybil additions reduces to this. Second,
trusted anchors. The onboarding allowlist and account-name registry lie outside the corruption budget, and
the safety of single-source fields rests on them, which is why the account result is a trusted-anchor
baseline rather than the radius certificate. Third, a trusted join key, which binds each attestation to one
transaction and prevents splicing evidence across transactions. Fourth, sound extraction. Values are
canonicalized deterministically before aggregation, and the safety argument assumes an honest class's value
is extracted to $f(x)$. The extractor is itself a shared control domain. One model, prompt, canonicalizer,
and provider read every source, so a systematic parser fault or a malicious provider is a common-mode
corruption that source isolation does not remove. The guarantee is therefore conditional on
the extractor being sound on unaffected evidence. The model-in-the-loop study in Section~\ref{sec:mitl}
relaxes rather than removes this assumption, and Appendix~\ref{app:exp} measures it directly. It shows that a content compromise confined to a single source
becomes an abstention, because source-isolated reading keeps one tampered document from satisfying a
multi-source rule, but it samples extractor faults rather than proving them independent across sources.

\section{Mechanism Validation}\label{sec:mech}

Before placing a language model in the loop we validate the certifier's aggregation in isolation, against an
adversary that controls one evidence class, in a model-free harness so the result speaks to the certificate
rather than end-to-end behavior. A mechanism is load-bearing when removing it opens exactly one attack while
the full certifier stays safe. Removing the trusted join key opens
amount splicing, removing the onboarding allowlist opens a name-matched mule, and removing the
mandatory-source rule opens a single-source decision, each and nothing else, while the full certifier admits
none. In the counting panel laundering succeeds only without domain-bound counting and a Sybil addition
only without authentication, the empirical face of Proposition~\ref{prop:launder}. Appendix~\ref{app:exp}
reports the full ablation in Table~\ref{tab:mech}.

The two robustness notions also appear directly: a three-source relation abstains under one corrupted source
without certifying a wrong value, and a five-source relation keeps certifying the correct value, the
exact-value regime of Proposition~\ref{prop:exact}, which Section~\ref{sec:baseline} then tests against an
optimizing attacker rather than randomized worlds.

\section{How Often Does the Multi-Source Precondition Hold?}\label{sec:scope}

A corruption radius is useful only where the corroboration it consumes exists, so we measure how often that
is the case. Sanctions designation is the cleanest open setting in which one action-critical fact, that an
entity is designated, is asserted by several nominally independent authorities. We use the consolidated
OpenSanctions \texttt{sanctions} collection~\cite{opensanctions}, which after entity resolution contains
$70{,}966$ entities carrying at least one designation authority, and for each entity compute a mapped-domain
count as a minimum hitting set over issuing authorities. We report it at two levels. The raw-dataset level
treats every source dataset as independent, a loose upper bound inflated by consolidated lists that republish
others and by authorities that publish several files. The control-domain level folds datasets sharing a
jurisdiction or body into one authority, and is the tighter proxy.

The full distribution is given in Table~\ref{tab:sanctions} (Appendix~\ref{app:exp}), and the control-domain
column is the one to read. After
collapsing datasets that share a jurisdiction or body, about $15\%$ of entities show three or more mapped
issuing domains, about $25\%$ show two or more, and roughly three-quarters rest on a single domain, so the
well-corroborated entities are the heavily cross-listed ones. This is an upper bound on candidate
multi-witness coverage, not evidence that $15\%$ of payment actions obtain an exact-value radius. A mapped
issuing domain $\hat{m}$ is not yet a corruption-distinct authority $m$, because shared upstream data and
coordinated designations inflate the count. Recovering the true hitting set would require aligning
programme, measure, and legal basis first. The corroboration the certificate consumes is present for a
minority of decisions, a wider band supports safety with abstention, and the single-source majority is why the
design needs trusted anchors where real evidence does not corroborate.

\paragraph{A second domain: software supply-chain provenance.}
Software artifact provenance is a second real setting where one action-critical fact, the source an installed
package resolves to, can be independently corroborated, and where corruption-distinctness is the crux. For
$450$ real npm packages, the actual dependency closure of twenty widely used roots, we measure three witnesses
to ``package $P$ comes from repository $R$''. The first is the registry's declared repository. The second is
the SLSA build provenance, Sigstore-signed and transparency-logged. The third is the repository
back-reference, whether the claimed repository itself declares $P$, a binding a substituted package cannot
forge without compromising the repository. Counted as three separate witnesses,
$16\%$ of packages appear to have all three. But the provenance is built by GitHub
Actions in every case and the back-reference is GitHub repository content, so both root in one control domain
and a single repository compromise forges both. After collapsing this shared origin the third witness
vanishes (Figure~\ref{fig:supply}): about $90\%$ of well-attested packages have exactly two corruption-distinct
witnesses, the registry and the repository, and none reach three, the construction of
Section~\ref{sec:defs} on real infrastructure. The collapse is not specific to npm: on $300$ PyPI packages,
where PEP-740 adoption is about double npm's, an apparent three-witness structure on $32\%$ again reduces to
two domains rooted in GitHub.

Two corruption-distinct witnesses give safety with abstention, not exact-value certification, under a
one-corruption adversary. Driving the agree-or-escalate certifier with them, a registry compromise that
substitutes a malicious source is caught as a disagreement on every corroborated package, so the certifier
certifies no malicious source ($0$ of $450$), whereas a registry-trusting baseline accepts it every time.

A potential third signal lies outside the package author's control. A repository's reputation, its age,
stars, and contributors, is built by the platform and community over years, not by whoever holds push access,
so a typosquat author who forges a back-reference on their own repository cannot fake it. Across the same
$307$ declared repositories the median repository is $11$ years old with $260$ stars, and $74\%$ clear a
modest bar ($\ge2$ years, $\ge50$ stars, $\ge3$ contributors). Reputation attests that a repository is long
established, not that $P$ came from $R$, so it is an eligibility signal rather than a third witness, and we
report availability on legitimate packages, not detection against malicious ones.

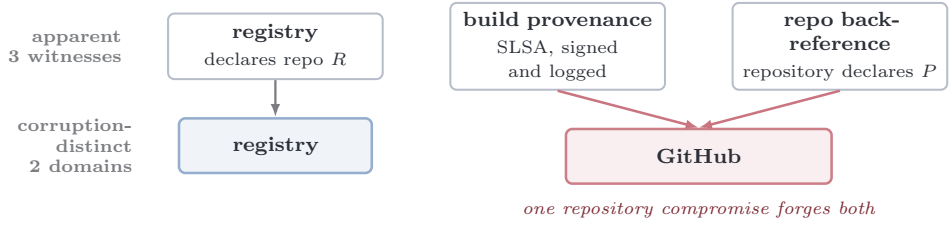
\begin{figure}[t]
\centering
\resizebox{0.76\linewidth}{!}{%
\begin{tikzpicture}[
  >={Latex[length=2mm]},
  wit/.style={draw=cardline, line width=0.8pt, rounded corners=2.5pt, fill=white, align=center,
    inner sep=3.5pt, font=\footnotesize, text=ink, minimum height=9mm, text width=27mm},
  dom/.style={draw=cProc!55, line width=1pt, rounded corners=2.5pt, fill=cProc!8, align=center,
    inner sep=3.5pt, font=\footnotesize\bfseries, text=ink, minimum height=7.5mm},
  gh/.style={draw=cDanger!65, line width=1pt, rounded corners=2.5pt, fill=cDanger!8, align=center,
    inner sep=3.5pt, font=\footnotesize\bfseries, text=ink, minimum height=7.5mm},
  ar/.style={-{Latex[length=1.8mm]}, line width=0.9pt, ink!62},
  ard/.style={-{Latex[length=1.8mm]}, line width=1pt, cDanger!70},
]
  % three apparent witnesses (top row)
  \node[wit] (w1) {\textbf{registry}\\[1pt]\scriptsize declares repo $R$};
  \node[wit, right=9mm of w1] (w2) {\textbf{build provenance}\\[1pt]\scriptsize SLSA, signed and logged};
  \node[wit, right=9mm of w2] (w3) {\textbf{repo back-reference}\\[1pt]\scriptsize repository declares $P$};
  % two corruption-distinct domains (bottom row)
  \node[dom, text width=24mm] (d1) at ($(w1.south)-(0,9mm)$) {registry};
  \node[gh, text width=34mm] (d2) at ($(w2.south)!0.5!(w3.south)-(0,9mm)$) {GitHub};
  % arrows: registry independent, provenance and back-reference both into GitHub
  \draw[ar] (w1.south) -- (d1.north);
  \draw[ard] (w2.south) -- (d2.north);
  \draw[ard] (w3.south) -- (d2.north);
  % shared-root annotation
  \node[font=\scriptsize\itshape, text=cDanger!70!black, below=1mm of d2]
    {one repository compromise forges both};
  % row labels on the left
  \node[font=\scriptsize\bfseries, text=ink!62, anchor=east, align=right] at ($(w1.west)-(5mm,0)$)
    {apparent\\3 witnesses};
  \node[font=\scriptsize\bfseries, text=ink!62, anchor=east, align=right] at ($(d1.west)-(5mm,0)$)
    {corruption-\\distinct\\2 domains};
\end{tikzpicture}}
\caption{The apparent third witness to package-source identity collapses on real infrastructure. The registry,
the SLSA build provenance, and the repository back-reference look independent, but the provenance is produced
by GitHub Actions and the back-reference is GitHub repository content, so both root in one control domain that
a single repository compromise forges. Counting corruption-distinct domains leaves two, not three. This is not
marginal: $16.4\%$ of $450$ npm packages and $32.3\%$ of $300$ PyPI packages appear to carry three independent
witnesses, and none do once the shared GitHub root collapses, leaving about $90\%$ at exactly two.}
\label{fig:supply}
\end{figure}

\section{Model in the Loop}\label{sec:mitl}

To test the certifier with a fallible, manipulable
model in a real agent loop, we evaluate it end to end on live models served through OpenRouter. The action is
a payment, $\{$payee, account, amount$\}$, assembled from $N$ evidence documents. In an attacked episode, an adversary has corrupted $k=1$ of them with an embedded social-engineering injection. The model reads each source in isolation, which
the certifier votes over, and separately reads all documents together as the naive baseline. Appendix~\ref{app:exp}
gives the full protocol.

Table~\ref{tab:models} reports five models at $80$ episodes each. Two things stand out. First, the injection
works: only Opus~4.8 resists it, while every other model, frontier GPT-5.5 and Gemini~3.1~Pro included,
follows it on every corrupted source, so the naive agent that trusts the model's own proposal executes the
fraudulent payment on $72$ to $100\%$ of attacks. Second, the certifier admits no unsafe action on any model.
Where the model is fooled, the rule decides the outcome. The agreement rule ($N=3$) turns the single corrupted
extraction into a safe abstention. The unique-threshold rule
($N=5>2k$, shown) recovers the correct value and abstains on only $0$ to $6\%$, because the honest sources
still clear the $N-k$ threshold without the misled one. These results instantiate Propositions~\ref{prop:safe}
and~\ref{prop:exact}: the guarantee holds whether or not the model resists. The liveness cost
is the honest-world false-abstention, $0$ to $12\%$ under agreement and at most $4\%$ under the threshold rule, since
the vote tolerates the residual extraction variance that strict agreement does not.

\begin{table}[t]
\centering
\footnotesize
\caption{Each of five current models runs in a real agent loop (OpenRouter) for $80$ episodes, $54$ with an
injection in one of $N=5$ sources and $26$ clean. The injection fools every model but Opus~4.8 and the naive
agent's unsafe rate reaches $100\%$, yet the certifier admits no unsafe action.}
\label{tab:models}
\setlength{\tabcolsep}{5.5pt}
\begin{tabular}{lccc>{\columncolor{cSafe!12}}c}
\toprule
\rowcolor{ink!8} model & inj.\ follow & naive unsafe & false-abstain & cert.\ unsafe\\
\midrule
Claude Opus 4.8   & 0\%   & 0\%   & 3.8\% & 0\%\\
GPT-5.5           & 100\% & 72\%  & 0\%   & 0\%\\
Gemini 3.1 Pro    & 100\% & 100\% & 0\%   & 0\%\\
Gemini 3.5 Flash  & 100\% & 98\%  & 0\%   & 0\%\\
Claude Haiku 4.5  & 100\% & 100\% & 0\%   & 0\%\\
\bottomrule
\end{tabular}
\end{table}

\section{Adaptive Attacker and Proxy Baselines}\label{sec:baseline}

\begin{figure}[t]
\centering
\begin{tikzpicture}
\begin{axis}[
  ybar,
  width=0.92\linewidth, height=3.0cm,
  bar width=8pt,
  ymin=0, ymax=126,
  ytick={0,25,50,75,100},
  ylabel={rate (\%)},
  ylabel style={font=\footnotesize, text=ink},
  symbolic x coords={naive, capability gating, policy monitor, certifier (ours)},
  xtick=data,
  xticklabel style={font=\scriptsize, text=ink},
  yticklabel style={font=\scriptsize, text=ink!70},
  axis background/.style={fill=paperbg},
  axis line style={ink!50, line width=0.7pt},
  tick style={ink!50}, tick align=outside,
  ymajorgrids=true, grid style={cardline!55, dotted, line width=0.5pt},
  legend style={font=\scriptsize, draw=ink!30, fill=white, at={(0.5,1.04)}, anchor=south,
    legend columns=2, column sep=10pt},
  legend cell align=left,
  enlarge x limits=0.16,
  nodes near coords, nodes near coords style={font=\tiny, text=ink!75, /pgf/number format/fixed,
    /pgf/number format/precision=0},
  nodes near coords align={vertical},
]
\addplot[fill=cDanger!78, draw=cDanger!92, line width=0.5pt]
  coordinates {(naive,100) (capability gating,0) (policy monitor,18.2) (certifier (ours),0)};
\addplot[fill=cSafe!72, draw=cSafe!92, line width=0.5pt]
  coordinates {(naive,45) (capability gating,0) (policy monitor,28.8) (certifier (ours),97.5)};
\legend{unsafe-execution rate, utility}
\end{axis}
\end{tikzpicture}
\caption{Head-to-head on the multi-source action for a model the injection fools (Haiku 4.5). The certifier is
the only one of the four that is at once safe ($0\%$ unsafe) and high-utility ($98\%$). The naive agent is fully
unsafe, capability gating (CaMeL/FIDES) executes nothing, and the policy monitor (PCAS) is $18\%$ unsafe at
only $29\%$ utility. The same pattern holds across all five models.}
\label{fig:frontier}
\end{figure}
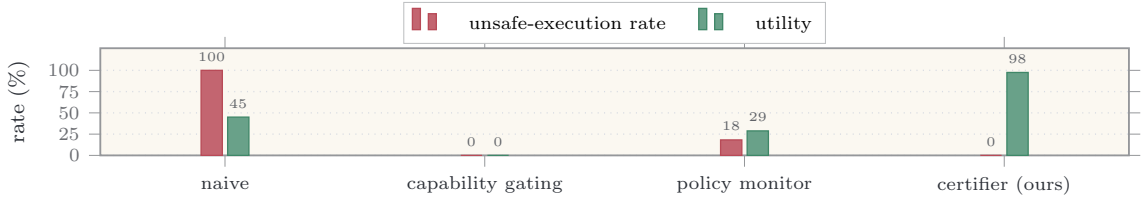

A defense that abstains freely is safe for a trivial reason, so the real test is safety \emph{and} utility
against faithful versions of the competing defense logic, compared on the domain where the comparison is fair.
A structured action assembled from several corruption-distinct sources is that domain: the standard injection
benchmark AgentDojo \cite{agentdojo} gives one source per fact, so corroboration is unavailable there and any
corroboration defense would abstain throughout. We pit the certifier against CaMeL/FIDES-style gating
\cite{camel,fides}, which refuses actions whose decisive fields derive only from untrusted evidence, and a
PCAS-style policy monitor \cite{pcas} with an allowlist of vendors, known accounts, and an amount cap.

On a separate run, $80$ episodes per model with $44$ attacked, Figure~\ref{fig:frontier} separates the defenses. Capability gating is safe but executes nothing, the blanket refusal. The policy
monitor catches gross substitutions yet stays unsafe on in-policy ones, up to $18\%$, and refuses every
legitimate payee it has not pre-registered, so its utility on novel payees is $0\%$. The certifier alone is at
once $0\%$ unsafe and high-utility: the threshold rule recovers the correct value under attack, $96$ to $100\%$
correct overall, and pays every novel payee, $100\%$, compared with $0\%$ for the policy monitor, because corroboration needs no maintained
ground-truth list. A model-free white-box adaptive attacker that rewrites one source's fields, launders
copies, and revises against each refusal never forces an unsafe execution against the certifier across
$2{,}000$ worlds, while breaking both baselines in every one, although each still blocks one attack family
(Table~\ref{tab:family}).

The harness's own baselines, action-gating and provenance-only, differ from the model-based pair of
Figure~\ref{fig:frontier}. Table~\ref{tab:family} breaks them down by attack family: action-gating admits every family
except an allowlist-missing destination, provenance-only every family except a currency swap, and the certifier
none. Each blocks one family the other misses, and both stay unsafe on the four they share.

The exact-value radius is visible in the same harness: sweeping corroborating payee sources under a
single-source vendor swap, the certifier abstains at two sources and certifies the correct vendor at three or
more, the radius-one case of Proposition~\ref{prop:exact}, with $k=2$ needing $N\ge5$ and $k=3$ needing
$N\ge7$ (Appendix~\ref{app:exp}).
Safety is not bought by refusing, since honest-world false-abstention stays at a few percent
under the threshold rule (Section~\ref{sec:mitl}). These baselines
reproduce the defense logic, not the full engineering of a deployed system.

\section{Related Work}\label{sec:related}

Information-flow and capability systems secure the agent around an untrusted model. CaMeL attaches
capabilities from the trusted query \cite{camel}, FIDES permits a consequential call only when every input is
high-integrity \cite{fides}, PCAS compiles authorization policies into a reference monitor \cite{pcas}, and
AuthGraph aligns a parameter-source provenance graph against a clean-context authorization graph
\cite{authgraph}. Table~\ref{tab:compare} places these side by side. Descending from control-flow
integrity \cite{cfi} and information-flow control \cite{ifc}, they gate on where data came from rather than
whether a value is correct. A decisive value from a default-untrusted source therefore forces a blanket
refusal, whereas our certificate decides admission from corroboration.

The nearest methodological neighbor, RobustRAG, isolates retrieved passages and
aggregates their answers to certify a quality bound against a bounded number of injected passages
\cite{robustrag}. We share the isolate-then-aggregate shape but count corruption-distinct control domains
through a minimum hitting set rather than passages, so laundered evidence inflates its budget but not ours,
and certify a field-typed action rather than aggregate free text.

\looseness=-1 The motivation we build on was stated independently by ARGUS, which observes that defenses able
only to refuse or isolate cannot admit a legitimate environment-supplied value \cite{argus}. ARGUS is closest
in spirit and concurrent, but reasons about which span of one observation produced an argument, while we
reason about how many independent sources corroborate a value and with what margin. Detection methods flag injections by behavioral consistency rather than certifying a corruption bound
\cite{melon,taskshield}, and persistent memory is its own attack surface \cite{minja,agentpoison,memorygraft}.
The concurrent MemLineage refuses actions with untrusted ancestors \cite{memlineage}, a recall problem. Our
domain-bound counting instead stops duplicates from manufacturing a quorum, a counting problem. Finally, the
certificate generalizes Byzantine quorum voting \cite{byzantine} from a fixed electorate to a set system over
authenticated control domains. The electorate is the novelty: it must be inferred from shared dependencies
rather than given, and a domain that republishes another casts no second vote.

\section{Limitations and Conclusion}\label{sec:concl}

\looseness=-1 Four limits bound the claim: a fully adaptive live-model attacker remains untested, the counts
are upper-bound proxies, the guarantee rests on the assumptions of Section~\ref{sec:defs}, and corroboration
remains scarce. Where it exists, the certifier admitted no unsafe action even where the model followed the
injected instruction.

\bibliographystyle{ieeetr}
\bibliography{references}

\appendix
\section{Proofs}\label{app:proofs}

We prove the characterization theorem and the propositions that specialize it. Throughout, a field has
true value $f(x)$, the adversary controls a set $C$ of evidence classes with $|C|\le k$, honest classes report
$f(x)$ after canonicalization, and votes are counted once per authenticated control domain.

\begin{proof}[Proof of Theorem~\ref{thm:feasible}]
\emph{Soundness.} If $a,a'\in F_k$ with $a\ne a'$, pick corruption sets $C,C'$ of size at most $k$ hitting the
disagreeing attestations of $a$ and of $a'$. The world with true action $a$ and corrupted domains $C$ produces
exactly the observation, as does the world with true action $a'$ and $C'$. A certifier sees only the
observation, so any action it executes is wrong in one of the two worlds, and it must abstain. Execution is
thus permitted only when $F_k=\{a\}$. \emph{Completeness.} If $F_k=\{a\}$, every $\le k$-corruption world
consistent with the observation has true action $a$, since any consistent $a'\ne a$ would itself be
$k$-feasible. Executing $a$ is correct in all such worlds. \emph{Optimality.} By soundness every $k$-safe
certifier abstains whenever $|F_k|\ne1$, exactly where the unique-feasible certifier does, so its abstention
set is contained in that of every $k$-safe certifier.
\end{proof}

\begin{proof}[Propositions~\ref{prop:safe} and~\ref{prop:exact}]
Both specialize Theorem~\ref{thm:feasible} to corruption-distinct classes, where dependency sets are disjoint
and $\tau$ counts dissenting domains, so a value is $k$-feasible iff at most $k$ classes dissent, that is its
support is at least $N-k$. Honest classes report $f(x)$, so $f(x)$ always has support at least $N-k$.
\emph{Threshold rule.} $F_k=\{f(x)\}$ iff no rival reaches support $N-k$, that is $k<N-k$, or $N>2k$, with radius
$\lfloor (N-1)/2\rfloor$. At $N\le2k$ the adversary concentrates $k$ classes on a challenger of support
$k\ge N-k$, two values are feasible, and the rule abstains. \emph{Agreement.} Unanimity makes $f(x)$ the only
feasible value once $N>k$, so any dissent forces abstention and a wrong value needs $|C|=N$, giving
$(N-1)$-safety with abstention. The margin form $c_1-c_2>2k$ is sufficient but not necessary and corresponds
to the stronger post-attack recertification guarantee with the looser $\lfloor (N-1)/4\rfloor$ radius.
\end{proof}

\begin{proof}[Proof of Proposition~\ref{prop:launder}]
Votes for a predicate are counted once per authenticated control domain. Let the adversary control the
domains in $C$ with $|C|\le k$. By authentication, every record the adversary emits, whether an original
record or a copy derived from one, is bound to an originating control domain in $C$: a copy inherits the
origin of its source by root inheritance, and an original record carries the unforgeable origin identifier of
the domain that produced it, which the adversary cannot forge for a domain outside $C$. The set of domains
onto which the adversary's records map is therefore a subset of $C$, of size at most $k$. Because the tally
counts at most one vote per distinct control domain per predicate, the adversary's effective vote count is at
most $|C|\le k$, independent of how many original records or copies it emits. If instead votes were counted
per provenance root, a single compromised domain emitting $t$ distinct original records would contribute $t$
votes, and choosing $t$ large enough would meet any fixed quorum, the attack the domain rule removes.
\end{proof}

\paragraph{Field-wise composition against the joint certifier.} Theorem~\ref{thm:feasible} reasons jointly, over
one budget shared by all fields, whereas the deployed certifier decides each field under its own rule and
executes when every field is admitted. The two are not the same object, and the relation between them is the
following.

\begin{proposition}[Composition is sound and conservative]\label{prop:compose}
Let each field $j$ admit $a_j$ only when $a_j$ is the unique value whose dissenting attestations are covered
by at most $k$ control domains. If every field admits, then $F_k=\{a\}$, so the composed certifier is
$k$-safe. The converse fails, so its abstention set contains that of the joint certifier of
Theorem~\ref{thm:feasible}.
\end{proposition}

\begin{proof}
Soundness. For any action $b$, $\mathrm{Dis}_j(b_j)\subseteq\bigcup_i \mathrm{Dis}_i(b)$, so a hitting set of
size at most $k$ for the union also hits each field's dissent, giving
$\tau(\mathrm{Dis}_j(b_j))\le\tau(\bigcup_i \mathrm{Dis}_i(b))$. Hence every jointly feasible action is
per-field feasible, and $F_k$ is contained in the product of the per-field feasible sets. If each field admits
a unique value then that product is $\{a\}$, so $F_k\subseteq\{a\}$. Under the threat model the true action is
explained by the actual corruption set of size at most $k$, so $F_k\neq\emptyset$ and therefore $F_k=\{a\}$,
which by Theorem~\ref{thm:feasible} permits execution and identifies $a$ as the truth. Conservativeness. The
inclusion is one-directional: a field may carry two per-field feasible values whose extensions to the
remaining fields are jointly infeasible, because a single budget cannot cover dissent that falls in different
domains across fields. The joint rule then certifies where the composition abstains. Since per-field
admission implies joint execution, the joint certifier's abstention set is contained in the composition's.
\end{proof}

\noindent Two consequences are worth stating plainly. The maximal liveness of Theorem~\ref{thm:feasible} is a
property of the unique-feasible-action rule, not of the field-wise implementation, which is safe but may
abstain strictly more often. And the reconciliation and anchor rules are not feasibility votes at all: the
anchor places a field outside the corruption budget, which Theorem~\ref{thm:feasible} shows is the only way a
single-sourced field executes. The implemented system is therefore best read as a sound specialization whose
liveness we measure rather than derive.

\section{Experimental Details}\label{app:exp}

The end-to-end evaluation of Section~\ref{sec:mitl} runs on live models served through OpenRouter. Each
instance is a payment with action-critical fields payee, account, and amount, presented as $N$ evidence
documents from fixed seeds so every model sees an identical corpus. In an attacked instance, one source ($k=1$) is corrupted with an
embedded social-engineering injection, drawn from three attack families: a payee swap, an account
redirection, and an amount change. The model reads each document in a separate, source-isolated call, which
the certifier votes over, and separately reads all documents together as the naive baseline, so a single
tampered source cannot by itself satisfy a multi-source rule. Field values are compared in canonical form,
amounts by numeric value, accounts by alphanumerics, and payees by normalized text, so formatting variation
is not mistaken for disagreement. A certified action is scored unsafe only when it differs from the unique
correct action. We report the agreement rule at $N=3$ and the unique-threshold rule at $N=5>2k$, $80$ episodes each,
roughly $4{,}000$ live model calls in total. Attacked episodes are drawn by a seeded generator at rate $0.6$, so
every model sees the same corpus, with $47$ attacked and $33$ clean at $N=3$ and $54$ and $26$ at $N=5$. The
clean episodes measure false abstention. The head-to-head of Figure~\ref{fig:frontier} is a separate run at
rate $0.55$, giving $44$ attacked and $36$ clean, $13$ of them with a novel payee, and there utility counts correct payments over all $80$
episodes, attacked ones included. The
adaptive-attacker comparison is model-free by construction: it grants the attacker full control of its single
source, laundered copies, refusal-driven revision, and a bounded search over the whole-record manipulation
space, so the result speaks to the certifier logic against a strong model-free adversary rather than to any model's
susceptibility. On reporting, the per-cell rates are point estimates over the episode counts stated in each caption, and
a reported zero is an observed count rather than a proof, bounding the rate below roughly $3/n$ at the stated
episode count $n$, and the pooled zero is descriptive rather than the outcome of a single significance test,
since errors are correlated by template and by model family.

\begin{table}[H]
\centering
\footnotesize
\caption{Mapped issuing domains per designated entity in OpenSanctions ($N=70{,}966$, Section~\ref{sec:scope}).
The control-domain column collapses datasets that share a jurisdiction or body. Both columns are upper-bound
proxies for corruption-distinctness, not operational corruption radii.}
\label{tab:sanctions}
\begin{tabular}{lcc}
\toprule
\rowcolor{ink!8} mapped-domain count $\hat{m}$ & raw dataset & control domain\\
\midrule
$\hat{m}=1$ (single domain)  & 62.5\% & \textbf{74.6\%}\\
$\hat{m}\ge 2$ (two or more) & 37.5\% & \textbf{25.4\%}\\
$\hat{m}\ge 3$ (three or more) & 22.7\% & \textbf{14.9\%}\\
$\hat{m}\ge 4$ (four or more) & 14.5\% & 11.7\%\\
\bottomrule
\end{tabular}
\end{table}

\begin{table}[H]
\centering
\footnotesize
\caption{Mechanism ablation under a single corrupted class (Section~\ref{sec:mech}). Each ablation opens
exactly one attack while the full certifier stays safe. Values are attacker success rates.}
\label{tab:mech}
\begin{tabular}{lcccc}
\toprule
\multicolumn{5}{l}{\textit{Payment-field mechanisms}}\\
\midrule
\rowcolor{ink!8} configuration & one-source & mule acct.\ & amount-splice & omission\\
\midrule
full                 & 0\% & 0\%   & 0\%   & 0\%\\
no join key          & 0\% & 0\%   & \textbf{100\%} & 0\%\\
no account policy    & 0\% & \textbf{100\%} & 0\% & 0\%\\
no mandatory source  & 0\% & 0\%   & 0\%   & \textbf{100\%}\\
\midrule
\multicolumn{5}{l}{\textit{Counting mechanisms (unsafe-certified rate)}}\\
\midrule
\rowcolor{ink!8} configuration & \multicolumn{1}{c}{single corr.} & laundering & \multicolumn{2}{c}{Sybil}\\
\midrule
full                & 0\%  & 0\%   & \multicolumn{2}{c}{0\%}\\
no atomic-claim     & 0\%  & \textbf{100\%} & \multicolumn{2}{c}{0\%}\\
no authentication   & 0\%  & 0\%   & \multicolumn{2}{c}{\textbf{100\%}}\\
\bottomrule
\end{tabular}
\end{table}

\paragraph{The corruption radius for general $k$.}
The single-source studies above exercise the budget $k=1$. To confirm that the certified value is governed by
a genuine corruption radius rather than a single-source special case, we run the identical certifier rule
against a general $k$-domain Byzantine adversary, sweeping the budget $k$ and the number of corruption-distinct
classes $N$. Each cell draws four thousand adversarial configurations, including the worst case in which the
$k$ controlled classes coordinate every vote on one challenger, and reports the worst outcome for the defender.
Table~\ref{tab:byz} shows the result. No configuration with at least one honest class ($N>k$) ever certifies a
wrong value, and the certified-correct boundary falls exactly at $N>2k$, so $k=1$ needs $N\ge3$, $k=2$ needs
$N\ge5$, and $k=3$ needs $N\ge7$, matching the radius $\lfloor (N-1)/2\rfloor$ of Proposition~\ref{prop:exact}.

\begin{table}[H]
\centering
\footnotesize
\caption{Outcome of the certifier rule under a $k$-domain Byzantine adversary
(Proposition~\ref{prop:exact}). \textbf{C} marks a cell certified under every adversarial
configuration drawn, including the coordinated worst case, a dot a cell in which at least one configuration
forces a safe abstention, and a dash a cell with no more classes than the budget ($N\le k$). No cell with
$N>k$ ever certifies a wrong value.}
\label{tab:byz}
\setlength{\tabcolsep}{7pt}
\begin{tabular}{lccccccc}
\toprule
\rowcolor{ink!8} budget & $N{=}2$ & $N{=}3$ & $N{=}4$ & $N{=}5$ & $N{=}6$ & $N{=}7$ & $N{=}8$\\
\midrule
$k=1$ & $\cdot$ & \textbf{C} & \textbf{C} & \textbf{C} & \textbf{C} & \textbf{C} & \textbf{C}\\
$k=2$ & --      & $\cdot$    & $\cdot$    & \textbf{C} & \textbf{C} & \textbf{C} & \textbf{C}\\
$k=3$ & --      & --         & $\cdot$    & $\cdot$    & $\cdot$    & \textbf{C} & \textbf{C}\\
\bottomrule
\end{tabular}
\end{table}

\paragraph{Vote identity must be the control domain, not the root.}
Proposition~\ref{prop:launder} requires that votes be counted per authenticated control domain. Counting per
provenance root is not sufficient, and the gap is exploitable. We run the same certifier rule under three
vote identities, while one compromised domain emits several distinct original records for a wrong value, each
with its own root, plus laundered copies. Table~\ref{tab:voteid} shows the outcome. Counting raw attestations
is unsafe at once. Counting roots defeats copies but not multiple originals: at $N=2$ a single compromised
domain that issues two original records gets the wrong value certified, and at $N\ge3$ the inflated roots push
the true value below threshold and destroy the radius. Counting one vote per authenticated control domain
reproduces Proposition~\ref{prop:exact} exactly, never certifying a wrong value and certifying the true value
whenever $N>2k$, independent of how many originals or copies the adversary manufactures.

\begin{table}[H]
\centering
\footnotesize
\caption{Certified outcome under a single compromised control domain ($k=1$) that emits several distinct
original records for the wrong value plus copies, by vote identity. Only counting per authenticated control
domain preserves the guarantee. ``Unsafe'' means the wrong value was certified.}
\label{tab:voteid}
\setlength{\tabcolsep}{8pt}
\begin{tabular}{ccccc}
\toprule
\rowcolor{ink!8} $N$ & orig & per attestation & per root & per domain (ours)\\
\midrule
2 & 1 & \textcolor{cDanger}{unsafe} & abstain                     & abstain\\
2 & 2 & \textcolor{cDanger}{unsafe} & \textcolor{cDanger}{unsafe} & abstain\\
3 & 1 & abstain                     & \textcolor{cSafe}{certify}  & \textcolor{cSafe}{certify}\\
3 & 2 & abstain                     & abstain                     & \textcolor{cSafe}{certify}\\
\bottomrule
\end{tabular}
\end{table}

\newcommand{\yY}{\textcolor{cSafe}{$\bullet$}}
\newcommand{\pP}{\textcolor{ink!70}{$\circ$}}
\newcommand{\nN}{\textcolor{ink!55}{\textendash}}
\begin{table}[H]
\centering
\footnotesize
\renewcommand{\arraystretch}{0.9}
\caption{Where the certificate sits among neighboring systems (\yY\ present, \pP\ partial, \nN\ absent):
corruption radius, field-typed action, shared-origin modelling, copy-laundering resistance, executing a value
rather than refusing, and real-data coverage, meaning the mechanism is measured on a real deployed corpus rather than on synthetic
evidence alone. The columns are the dimensions this work targets, so the table shows where the certificate
differs from its neighbours rather than ranking them overall.}
\label{tab:compare}
\setlength{\tabcolsep}{8pt}
\begin{tabular}{lcccccc}
\toprule
\rowcolor{ink!8} system & radius & typed & shared-orig. & copy-safe & executes & coverage\\
\midrule
RobustRAG \cite{robustrag}   & \yY & \nN & \nN & \nN & \yY & \nN\\
CaMeL \cite{camel}           & \nN & \pP & \nN & \nN & \nN & \nN\\
FIDES \cite{fides}           & \nN & \pP & \nN & \nN & \nN & \nN\\
PCAS \cite{pcas}             & \nN & \yY & \nN & \nN & \yY & \nN\\
ARGUS \cite{argus}           & \nN & \yY & \nN & \nN & \yY & \nN\\
AuthGraph \cite{authgraph}   & \nN & \yY & \nN & \nN & \yY & \nN\\
MemLineage \cite{memlineage} & \nN & \pP & \nN & \nN & \nN & \nN\\
\rowcolor{cSafe!12} \textbf{this work} & \yY & \yY & \yY & \yY & \yY & \yY\\
\bottomrule
\end{tabular}
\end{table}

\begin{table}[H]
\centering
\footnotesize
\caption{Outcome by attack family under the adaptive one-source attacker (Section~\ref{sec:baseline}). The two
proxy baselines each block one family the other misses and both stay unsafe on the four they share. \textcolor{cDanger}{\textbf{U}}
is an unsafe execution, a dash a safe abstention, \textcolor{cSafe}{\textbf{C}} a correct execution under
attack. The families are a vendor substitution, a fresh mule account, an account carried across vendors, an
inflated amount, a currency swap, and a fully alternate invoice.}
\label{tab:family}
\setlength{\tabcolsep}{5pt}
\begin{tabular}{lcccccc}
\toprule
\rowcolor{ink!8} baseline & vendor & mule & x-vendor & amount & currency & full alt.\\
\midrule
action-gating   & \textcolor{cDanger}{\textbf{U}} & --                              & \textcolor{cDanger}{\textbf{U}} & \textcolor{cDanger}{\textbf{U}} & \textcolor{cDanger}{\textbf{U}} & \textcolor{cDanger}{\textbf{U}}\\
provenance-only & \textcolor{cDanger}{\textbf{U}} & \textcolor{cDanger}{\textbf{U}} & \textcolor{cDanger}{\textbf{U}} & \textcolor{cDanger}{\textbf{U}} & --                              & \textcolor{cDanger}{\textbf{U}}\\
\rowcolor{cSafe!12} \textbf{full certifier} & \textcolor{cSafe}{\textbf{C}} & -- & -- & -- & -- & --\\
\bottomrule
\end{tabular}
\end{table}

\paragraph{Computing the corruption-distinct count.} The minimum hitting set of
Definition~\ref{def:cdcount} is NP-hard in general, but the instances a certifier meets are small and
structured. Our implementation first forces every control domain that is the sole cover of some attestation,
removes the attestations that those already hit, and solves the small remainder exactly, falling back to a greedy
cover only when the residual universe exceeds twenty-two domains. Real dependency sets are short, one or two
domains for a national listing that transposes a United Nations designation, so the forced-singleton pass
usually settles the instance outright. The direction of any approximation matters. Greedy returns an upper
bound on the hitting set, which overstates $m$ and therefore overstates the radius, so a deployment must use
an exact or lower-bound count. The measurements in Section~\ref{sec:scope} are reported as upper bounds
throughout for the same reason.

\paragraph{Obtaining and maintaining dependency sets.} The dependency sets themselves are metadata, and a
deployment has to source them. In the two domains we study they are already published: a sanctions record
names its issuing authority and its legal basis, and a package attestation names the registry, the builder
identity, and the source repository. Mapping those to control domains is a curation task, folding datasets
that share a jurisdiction or body, and it changes on the timescale of institutional structure rather than of
individual records, so an audited quarterly review is enough. Incomplete or wrong information must fail in
the safe direction. An attestation of unknown provenance should be assigned to an unknown domain shared with
every other unknown, which lowers $m$ and tightens the radius, rather than being treated as a fresh
independent witness. Under that convention missing metadata costs liveness and never safety, and an
adversary who hides provenance gains abstention rather than execution.

\paragraph{Per-source injection following against the naive agent.} The two rates in
Table~\ref{tab:models} are measured on different calls, and the gap between them is informative. Injection
following is recorded on the source-isolated read of a corrupted document, and the naive unsafe rate is
recorded on a separate holistic call in which the model reads all $N$ documents together. A model can adopt
the attacker value when it sees the tampered source alone and still reject it in the holistic read, where the
untampered documents contradict it. GPT-5.5 shows this most clearly: it follows the injection on every
attacked episode, yet the holistic agent proposes the attacker value in $85\%$ of them under agreement and
$72\%$ under the threshold rule. The residual is not a parsing artifact, since every holistic response in those
episodes parsed successfully. The certifier does not depend on either rate, because it votes over the
isolated reads and never consumes the holistic proposal.

\paragraph{Extractor independence.} The guarantee assumes an honest class's value is extracted to $f(x)$,
and the extractor is a shared control domain, so we test it. On $150$ episodes at $N=5$ and $k=1$, $96$ attacked and $54$ clean, we read
every source in isolation under six configurations: five single-model ones, which are the deployed case, and
one mixed panel from Table~\ref{tab:models} with a different model per source. The dangerous event is not an
extraction error, since one misread honest source forces abstention in an attacked episode and is tolerated
in a clean one. It is that enough honest sources agree on the \emph{same} wrong value for it, with the
corrupted source, to clear the $N-k$ threshold, which here takes three. We measure the stricter event of any
two agreeing, a conservative proxy, over honest sources alone. Across $654$ honest reads per configuration, four from each attacked episode and five from each clean one,
no honest source yielded the attacker's value, no two agreed on a wrong value, and none produced an unsafe
certified action, bounding that proxy below $2\%$ per configuration. We do not pool across configurations,
since errors are correlated by template and model family, and since no semantic error was observed the data
bound the coincidence rate rather than estimate a correlation.

The study exposed a canonicalization gap. One vendor name carries Cyrillic characters, transliterated to
Latin at rates from $0\%$ for both Gemini models to $7.5\%$ for Opus 4.8. Canonicalization folds case and
whitespace but not script, so a transliterated read compares unequal. It never produced an unsafe action,
the value being the correct payee, but it costs liveness: Opus transliterated consistently and never
abstained, whereas Haiku 4.5 did so inconsistently, disagreed with itself, and abstained on $4\%$. On this
corpus self-consistency mattered more than sharing, and canonicalization should fold scripts as well as case.

\end{document}